\documentclass[conference]{IEEEtran}
\IEEEoverridecommandlockouts
\usepackage{amsmath,amssymb,amsfonts,amsthm}
\usepackage{mathtools}
\newtheorem{theorem}{Proposition}[]
\usepackage{bm}
\usepackage{dsfont}
\usepackage{physics}
\usepackage{algorithmic}
\usepackage{graphicx}
\usepackage{textcomp}
\usepackage{xcolor}
\usepackage{orcidlink}

\def\BibTeX{{\rm B\kern-.05em{\sc i\kern-.025em b}\kern-.08em
    T\kern-.1667em\lower.7ex\hbox{E}\kern-.125emX}}

\newcommand{\R}[0]{\mathbb{R}}

\newcommand{\SComplex}[0]{\mathcal{S}}

\newcommand{\Hilbert}[0]{\mathcal{H}}

\newcommand{\filtration}[0]{\mathfrak{F}}
    
\begin{document}

\title{Higher-order topological kernels \\ via quantum computation
\thanks{MI and APD are partially supported by the Istituto Nazionale di Alta Matematica ``Francesco Severi''. FM is supported by the grant AdR4029/22 cofinanced by TAS Group and the University of Verona.}
}

\author{
\IEEEauthorblockN{Massimiliano Incudini\orcidlink{0000-0002-9389-5370}}
\IEEEauthorblockA{\textit{Department of Computer Science} \\
\textit{University of Verona}\\
Verona, Italy \\
massimiliano.incudini@univr.it}
\and
\IEEEauthorblockN{Francesco Martini\orcidlink{0000-0003-2651-140X}}
\IEEEauthorblockA{\textit{Department of Computer Science} \\
\textit{University of Verona}\\
Verona, Italy \\
francesco.martini@univr.it}
\and
\IEEEauthorblockN{Alessandra Di Pierro\orcidlink{0000-0003-4173-7941}}
\IEEEauthorblockA{\textit{Department of Computer Science} \\
\textit{University of Verona}\\
Verona, Italy \\
alessandra.dipierro@univr.it}
}

\maketitle

\begin{abstract}
Topological data analysis (TDA) has emerged as a powerful tool for extracting meaningful insights from complex data. TDA enhances the analysis of objects by embedding them into a simplicial complex and extracting useful global properties such as the Betti numbers, i.e. the number of multidimensional holes. which can be used to define kernel methods that are easily integrated with existing machine-learning algorithms. These kernel methods have found broad applications, as they rely on powerful mathematical frameworks which provide theoretical guarantees on their performance. However, the computation of higher-dimensional Betti numbers can be prohibitively expensive on classical hardware, while quantum algorithms can approximate them in polynomial time in the instance size. In this work, we propose a quantum approach to defining topological kernels, which is based on constructing Betti curves, i.e. topological fingerprint of filtrations with increasing order. We exhibit a working prototype of our approach implemented on a noiseless simulator and show its robustness by means of some empirical results suggesting that topological approaches may offer an advantage in quantum machine learning.
\end{abstract}


\section{Introduction}

Quantum computing holds the promise of revolutionizing the field of machine learning due to the increase in computation power \cite{biamonte2017quantum}. Despite the numerous efforts made by researchers, no clear direction for achieving any advantage has been discovered yet. 
One of the initial techniques developed is the quantum kernel, which computes a kernel function via a quantum device \cite{havlivcek2019supervised}. Its success is attributed to the powerful mathematical framework of kernel methods, which is based on functional analysis and statistical learning and provides theoretical guarantee on their performance
\cite{canatar2021spectral}. The most basic type of quantum kernel essentially consists of a quantum embedding, a parametric unitary function that encodes classical data into rotational angles of quantum gates. The embedded data is processed using the fidelity test or the SWAP test to estimate the inner product. This kernel family has gained popularity in the quantum realm due to its immediate applicability to the current generation of devices \cite{preskill2018quantum}. This approach demonstrates superior performance compared to conventional kernels in certain physics-related practical applications \cite{wozniak2023quantum}. However, its use is hindered by unfavorable properties such as the flat distribution of exponentially small eigenvalues \cite{kubler2021inductive}, which prevents the learning of associated components, and the exponential concentration of coefficients, which makes kernel calls indistinguishable \cite{thanasilp2022exponential}.

An alternative method of constructing quantum kernels involves the use of fault-tolerant hardware. In \cite{liu2021rigorous}, a quantum kernel is proposed that uses the Shor algorithm and applies to problems involving single-feature data $x \in \mathbb{Z}_p$ that becomes linearly classifiable, i.e. trivial to classify, after being transformed via the mapping $x \mapsto \log_g(x)$ with $g \in \mathbb{Z}_p$. This approach allows for an exponential speedup of the quantum algorithm over the classical one, given widely accepted cryptographic assumptions. While this technique provides strong guarantees for this particular problem, it is highly artificial and not applicable to most real-world problems. Following a similar approach, in this paper we suggest the use of the Lloyd-Garnerone-Zanardi (LGZ) \cite{lloyd2016quantum} algorithm to estimate the topological feature of complex data for the construction of quantum kernels which offer both theoretical guarantees and practical applicability to real-world tasks. The LGZ is exponentially faster than the best classical method identified to date, and no attempts to invalidate this advantage via dequantization techniques have succeeded \cite{chia2022sampling}. While classical techniques are efficient for computing topological features of lower dimensions, the quantum framework seems to be indispensable for higher-dimensional ones. Fig. \ref{fig:tda-pipeline}(a) illustrates the scenario where quantum hardware can outperform its classical counterpart.

\begin{figure*}[tb]
    \centering
    \includegraphics[width=0.95\textwidth]{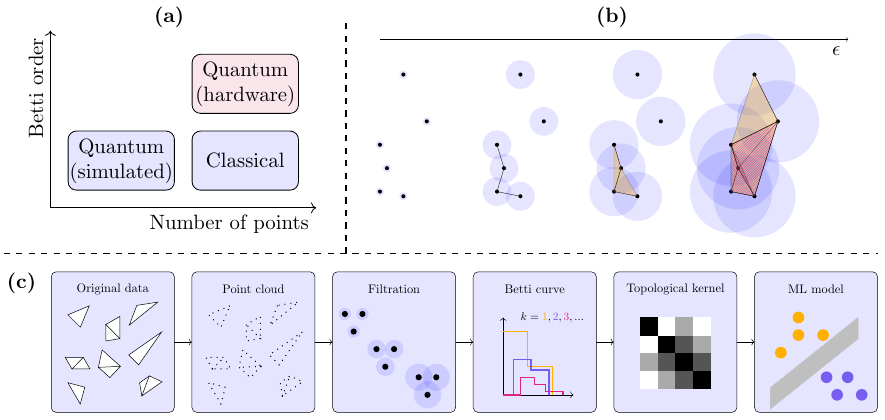}
    \caption{Key aspects of quantum topological data analysis and its application in the present work. \textbf{(a)} The feasibility of TDA for large point clouds on classical hardware is limited to lower-order Betti numbers, while efficient analysis of higher-order Betti numbers is expected on quantum hardware, assuming large-scale fault-tolerant quantum hardware is available. The quantum algorithm is not simulable on classical hardware for anything other than trivial examples with few points. \textbf{(b)} Construction of a Vietoris-Rips filtration. The features are first represented as a point cloud using an embedding function. Then, the skeleton of the simplicial complex is defined as a graph with vertices representing the points, and edges connecting pairs of points whose distance is less than or equal to $\epsilon$. Next, $k$-cliques are identified and added to the simplicial complex, possibly up to a maximum value of $k_\text{max}$. Finally, the filtration is defined by a sequence of simplicial complexes, $\SComplex_{1} \subset ... \subset \SComplex_{n}$ for $\epsilon_1 < ... < \epsilon_n$. \textbf{(c)} Pipeline integrating the topological kernel calculated with quantum computation into the machine learning model. The original data is embedded to a point cloud, used to construct the filtration, followed by quantum hardware computing Betti numbers for varying $\epsilon$ and order $k$ to create Betti curves. These curves are then embedded into a Hilbert space to define a positive definite kernel function that is integrated into the machine learning model.}
    \label{fig:tda-pipeline}
\end{figure*}

In the past decades, Topological Data Analysis (TDA) \cite{chazal2021introduction}, i.e. the extraction of useful topological features from data, has been widely studied.
In TDA, data is typically transformed into a point cloud, and a skeleton graph is generated by connecting vertices whose Euclidean distance is less than a threshold value $\epsilon$. The next step involves creating an abstract simplicial complex (ABS) from this skeleton graph, which captures the higher-order relationships between points (simplices). One common type of ABS is the Vietoris-Rips complex, which reveals all the $k$-cliques in the graph, potentially up to a certain order. By analyzing the ABS, we can determine the Betti number of order $k$, denoted $\beta_k$, which represents the number of $k$-dimensional holes in the structure. Additionally, a \emph{filtration} can be created by ordering a collection of ABS based on increasing values of $\epsilon$, resulting in a sequence of nested structures. \emph{Persistent Betti numbers} can then be calculated from the filtration, tracking the birth and death of each multidimensional hole, i.e., the minimum threshold for which the hole appears and the maximum threshold after which it disappears. Fig. \ref{fig:tda-pipeline}(b) provides an illustration of the filtration construction process.

Topological Data Analysis has the potential to enhance various aspects of machine learning, such as data exploration, feature engineering, and visual representation \cite{hensel2021survey}. In the context of feature engineering, TDA provides an alternative approach to traditional techniques that rely on distance measures in vector spaces. However, the integration process requires particular care. A topological fingerprint of the data (e.g. vector of Betti numbers) can be extracted and fed as additional features to most machine learning models, such as neural networks. Alternatively, the same aspects can be included in the definition of a kernel function, allowing the use of this powerful mathematical framework and integration with most distance-based machine-learning algorithms. Authors in \cite{reininghaus2015stable} have proposed a kernel technique based on the \emph{persistence diagram}, a two-dimensional representation of the topological aspects of a filtration, in which each point $(b_i, d_i)$ corresponds to the birth and death of the $i$-th multidimensional hole. Such representation can be embedded into a Hilbert space, whose inner product allows the construction of the similarity measure. Furthermore, it can be proven that this approach is stable with respect to the 1-Wasserstein distance, i.e. robust to small perturbations of the dataset \cite{carriere2017sliced}. An alternative definition of the kernel can be based on a more succinct representation of the topological features, the \emph{Betti curves}, introduced in \cite{umeda2017time}. These consider only the total number of $k$-dimensional holes for a given filtration, without the need to compute persistent Betti numbers. It is possible to obtain a Betti curve given the corresponding persistence diagram, although such a transformation is not injective. Finally, defining a distance (or similarity) measure between two Betti curves is straightforward because they are essentially bounded truncated piecewise linear functions.

In this work, we incorporate the data obtained from the LGZ algorithm into a kernel machine, show its efficacy, and explore the possibility of using our method in large-scale applications. 
To accomplish this, we introduce the concept of multidimensional Betti curves, which extend the Betti curves to explicitly consider various degrees of Betti numbers. This technique enables us to leverage the topological information provided by the LGZ algorithm without relying on the persistent Betti numbers. We demonstrate how multidimensional Betti curves can be used to define a kernel function and implement our approach using Python and the Qiskit platform. To evaluate the performance of our approach, we conduct experiments on the shape classification problem, a well-known benchmark. We compare the accuracy of a kernel machine with the higher-order topological kernel to conventional kernel functions, including the Gaussian kernel and polynomial kernels in various configurations. We calculate the topological kernels using either the classical exact procedure or by executing the LGZ algorithm on a noiseless quantum simulator to imply that such an algorithm requires fault-tolerant hardware. We also conduct further experiments to demonstrate the approach's robustness to errors introduced by the Hamiltonian simulation and finite sampling. In most cases, the quantum (approximated) technique closely resembles the classical (exact, but computationally expensive) method, even when using stochastic Hamiltonian simulation techniques such as qDrift.  Finally, we discuss potential applications based on time series analysis that could benefit from higher-degree Betti numbers and the advantages that quantum techniques may provide.

\section{Background}

This section provides a brief overview of the relevant background for this work.

\subsection{Kernel methods}

The application of kernel methods in both supervised and unsupervised machine learning is widespread and essential in these domains. These methods are described in detail in various sources, including \cite{steinwart2008support}. Kernel methods can be integrated with most similarity-based algorithms, including Ridge regression and Support Vector Machine, Principal Component Analysis, k-means, and DBSCAN. In this section, we present a concise summary of these ideas.

A \emph{kernel function} $\kappa : X \times X \to \R$ extends the concept of similarity for a non-empty set $X$. If $X$ is an inner product space, the inner product can be used as a similarity measure for a distance-based machine learning algorithm. However, if the inner product fails to capture the relationship between the data points or is unavailable, a kernel function can be used, which maps elements of $X$ into a higher-dimensional Hilbert space $\Hilbert$ using a feature map $\phi: X \to \Hilbert$,
\begin{equation}\label{eq:def_kernel}
    \kappa(x_1, x_2) = \langle \phi(x_1), \phi(x_2) \rangle_\Hilbert.
\end{equation}
Equivalently, $\kappa$ is a kernel if it satisfies positive definiteness and, by Mercer's theorem, this guarantees that there exists a feature map  $\phi$ and Hilbert space $\Hilbert$ satisfying Equation (\ref{eq:def_kernel}).

Considering the supervised learning problem
\begin{equation}
    f^* = \arg\min_{f \in \Hilbert} \sum_{i=1}^m \ell(f(x^{(i)}), y^{(i)}) + \lambda \lVert f \rVert_\Hilbert^2
\end{equation}
where $\{ (x^{(i)}, y^{(i)}) \}_{i=1}^m$ is the training set, $\ell$ convex loss function and $\lambda > 0$, the Representer Theorem guarantees that the learning problem is convex, $m$-dimensional and its solution can be expressed as $f^*(x) = \sum_{i=1}^m \alpha_i \kappa(x^{(i)}, x)$ with $\alpha_i$ parameters of the model. This contrasts with neural networks and most machine learning models, whose training is non-convex and provides no guarantee of finding the optimal solution.

\subsection{Topological data analysis}

We provide the necessary background in algebraic topology to comprehend the quantum algorithm. For a more detailed overview, we refer readers to \cite{chazal2021introduction,rieck2020topological}.

A \emph{$k$-simplex} $\sigma = \sum_{i=0}^k \alpha_i p_i$ is the convex hull of $k+1$ independent points $p_0, ..., p_k \in \R^d$, with $\alpha_i > 0$ and $\sum_{i=0}^k \alpha_i = 1$. The \emph{dimension} of $\sigma$ is $k$. A face of $\sigma$ is the convex hull of a subset of $\sigma$'s points. A \emph{simplicial complex} $S$ is a set of simplices such that $\sigma \in S, \tau$ is a face of $\sigma$ implies $\tau \in S$, and given $\sigma, \tau \in S$, their intersection $\sigma \cap \tau$ is either the empty set or a face of both.

The simplicial complex $S$ has a corresponding purely combinatorial entity known as the \emph{abstract simplicial complex} $\SComplex = (V, \Sigma)$ where $V = \{0, ..., k\}$ is the set of vertices and $\Sigma$ is a collection of subsets of $V$ such that $\sigma \in \Sigma$ and $\varnothing \neq \tau \subseteq \sigma$ implies $\tau \in \Sigma$. The elements of $\Sigma$ are called simplices. 

A \emph{$p$-chain} $c$ is a sum of $p$-simplices $\sigma_i \in \SComplex$, each multiplied by a coefficient $\alpha_i$. In this context, the coefficients are taken from the ring $\mathbb{Z}_2$, and chains can be treated as sets. The collection of $p$-chains, with addition modulo two, forms the chain group $C_p(\SComplex)$ or simply $C_p$. For a $p$-simplex $\sigma$, the \emph{boundary operator} $\partial_p: C_p \to C_{p-1}$ is given by
\begin{equation}
\partial_p \sigma = \sum_{j=0}^k (-1)^j \{ v_0, ..., v_k \} \setminus \{ v_j \}.
\end{equation}
It holds that $\partial_{p-1} \circ \partial_p = 0$. A $p$-chain $c$ is called a $p$-cycle if $\partial c = 0$. The collection of all $p$-cycles forms a group under addition modulo two, known as the \emph{cycle group} $Z_p$. 
Applying the boundary operator $\partial_{p+1}$ to all $p+1$-chains in $C_{p+1}$ produces the \emph{boundary group} $B_p$. The $p$-th  \emph{homology group} is the quotient group $Z_p / B_p$, which is also a vector space of dimension $\beta_p = \mathrm{dim} H_p$. 
The number $\beta_p$ is called the \emph{$p$-th Betti number}. In addition, there is a relationship between $H_p$ and the Hodge Laplacian $\Delta_p = \partial_{p}^\dagger \partial_{p} + \partial_{p+1} \partial_{p+1}^\dagger$ provided by Hodge theory, $H_p = \mathrm{ker} \Delta_k$.

A filtration of a simplicial complex $\SComplex$, $\filtration(\SComplex)$ or simply $\filtration$, is a nested sequence of its subcomplex $\filtration: \varnothing = \SComplex_0 \subseteq ... \subseteq \SComplex_n = \SComplex$. Every filtration gives rise to a sequence of homomorphisms $h^{i,j}_p : H_p(\SComplex_i) \to H_p(\SComplex_j)$. The $p$-th persistent homology group is $H_p^{i,j} = \mathrm{im} \, h^{i,j}_p$ and $\mathrm{dim} H_p^{i,j}$ is the \emph{$p$-th persistent Betti number}.

Given a metric space with distance $d$ and $r > 0$, the abstract simplicial complex $\SComplex$ is a \emph{Vietoris-Rips complex} when $\sigma \in \SComplex$ if and only if $d(p, q) \le 2r$ for all vertices $p, q \in \sigma$. In the case of Vietoris-Rips complexes, we can represent $\SComplex$ using a graph $G = (V, E)$, with vertices representing $0$-simplices, edges representing $1$-simplices, and $(k+1)$-cliques representing $k$-simplices.

\section{Quantum computation of Betti numbers}\label{sec:quantum}

\begin{figure*}[t]
    \centering
    \includegraphics[width=0.90\textwidth]{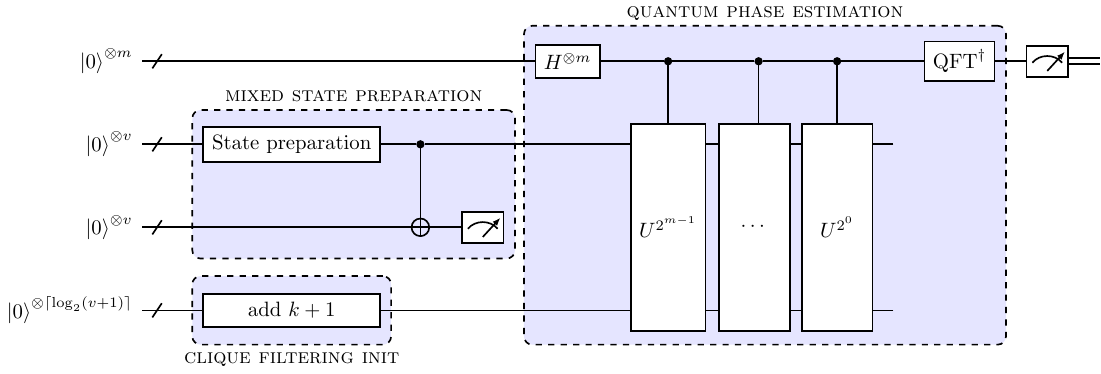}
    \caption{Quantum circuit implementation of the LGZ algorithm.}
    \label{fig:lgz-circuit}
\end{figure*}

The authors in \cite{lloyd2016quantum} proposed a quantum algorithm for approximating $\beta_k$ of a simplicial complex $\SComplex$. In this section, we briefly discuss this algorithm and explore both its potential benefits and limitations.

\subsection{LGZ algorithm}

The LGZ algorithm takes in a Vietoris-Rips complex $\SComplex$ (or equivalently its graph $G$) and a non-negative integer $k$, and returns an approximation of the Betti number $\beta_k$. The number of vertices of $G=(V, E)$ is denoted with $|V| = n$. We adopt the notation from \cite{gyurik2022towards} by using $Cl_{k}(G) \subset \{ 0, 1 \}^n$ to denote the set of $k+1$-clique in $G$ of binary strings having Hamming weight $k+1$ and whose onset corresponds to the vertices of the clique. The algorithm works as follows. 

Firstly, we prepare the state
\begin{equation}
\rho_k^G = \frac{1}{|Cl_{k}(G)|} \sum_{j \in Cl_{k}(G)} \ketbra{j}{j}.
\end{equation}
This is achieved by encoding the uniform superposition of $k+1$-cliques in $G$ using an $n$-qubits register, which can be accomplished either by utilizing Grover's algorithm as outlined in \cite{gunn2019review} or by employing a state preparation algorithm \cite{ameneyro2022quantum}. While Grover's algorithm can be useful for preparing states in general, the overhead it introduces may make the use of a state preparation routine more convenient for proof-of-concept implementations with small-scale systems. Next, a CNOT gate is applied with the control qubit on each qubit of the state to a separate ancilla to obtain the mixed state, and the $v$ ancillae are then measured and their results discarded.

We now perform quantum phase estimation over the eigenvector $\rho_k^G$ and unitary operator  $\exp(-i \Delta)$, with $\Delta$ the Hodge-Laplacian operator. Note that, instead of simulating $\Delta$, we can instead use the Dirac operator $B$,
\begin{equation}
    B =
    \scalebox{0.7}{$\begin{bmatrix}
    0 & \partial_{1}^\SComplex & & & & \\
    (\partial_{1}^\SComplex)^\dagger & 0 & \partial_{2}^\SComplex & & & \\
    & (\partial_{2}^\SComplex)^\dagger & 0 & & & \\
    & &  & \ddots & & \\
    & & & & 0 & \partial_{n-1}^\SComplex \\
    & & & & (\partial_{n-1}^\SComplex)^\dagger & 0
    \end{bmatrix}$},
\end{equation}
which is the square root of the Laplacian operator and is $n$-sparse thus easier to implement. As we are interested in $\mathrm{ker} \Delta_k^G$ rather than $\mathrm{ker} \Delta$, we modify the routine by evolving the operator $I \otimes B$ on the state $\ketbra{k+1}{k+1} \otimes \rho_k^\SComplex$ using $\lceil \log_2(n+1) \rceil$ more qubits. This ensures that the estimated eigenvalue is nonzero for $\ket{j} \not \in \mathrm{ker} \Delta_k^G$ \cite{gunn2019review}. Furthermore, to prevent the estimation of eigenvalues that are multiples of $2\pi$, we can rescale $B$ by the factor $\lambda_\text{max}^{-1}$, where $\lambda_\text{max}$ is the highest eigenvalue of $B$ which is bounded by $O(n)$ due to the Gershgorin circle theorem. The operator $\exp(-i B)$ can be implemented using a Hamiltonian simulation algorithm, such as Trotter decomposition \cite{childs2021theory} or stochastic techniques like qDrift \cite{campbell2019random}. These approaches provide varying levels of resource requirements and error rates, with stochastic techniques being less resource-intensive but producing a larger error. The use of importance sampling in the stochastic selection of the operator terms \cite{kiss2023importance} can lead to further resource reduction. 

Finally, we perform the estimation process $M \in O(\epsilon^{-2})$ times to achieve a precision of $\epsilon$. We use the obtained phases $\theta^{(1)}, ..., \theta^{(M)}$ to compute
\begin{equation}
\beta_k \approx \tilde{\beta}_k = \frac{|\{ \theta^{(i)} \mid \theta^{(i)} = 0 \}|}{M} \times |Cl_{k}(G)|.
\end{equation}

To distinguish the phase $0$ from the smallest eigenvalue we require a precision of $\kappa = \lambda_\text{max} / \lambda_\text{min}$.
However, no bound is known for $\lambda_\text{min}$. Nonetheless, if we fix a precision that cannot distinguish $\lambda_\text{min}$ from $0$, LGZ will still provide useful topological information \cite{gyurik2022towards}, even though such information does not correspond to the normalized Betti numbers. The circuit is shown in Fig. \ref{fig:lgz-circuit}.

\subsection{Limitations of the LGZ algorithm}

While the LGZ algorithm has the potential to achieve a superpolynomial speedup, it is crucial to address its limitations.
Clearly, this algorithm requires fault-tolerant quantum computers, although a cheaper variation that can be executed on NISQ devices exists \cite{akhalwaya2022exponential}.

According to \cite{schmidhuber2022complexity}, we can estimate the normalized Betti number $\beta_k / |Cl_k(G)|$ up to an additive error $\varepsilon$ in time 
\begin{equation}
    O\left(\frac{n^3 \kappa + n k^2 \zeta_k^{-1/2}}{\varepsilon^2}\right).
\end{equation}
Firstly, the state preparation is performed via Grover's algorithm using the oracle
\begin{equation}
    f_k(j) 
    = \begin{cases} 1, & j \in Cl_k(G) \\ 0, & j \not\in Cl_k(G)\end{cases},
\end{equation}
which can be implemented in $O(k^2)$, and requiring $O(\zeta^{-1/2}_k)$ oracle calls, where $\zeta_k = |Cl_k(G)|/\binom{n}{k+1}$. The cost of applying CNOTs for constructing the mixed state is $O(n)$. Secondly, quantum phase estimation of an $n$-sparse operator requires $O(n^3)$ gates and a precision of $O(\kappa)$, where $\kappa = \lambda_\text{max}/\lambda_\text{min}$. Thirdly, the cost of sampling is $\varepsilon^{-2}$.

Imposing a multiplicative error of $\delta$ on the \emph{actual} Betti numbers, obtained fixing $\varepsilon = \delta \beta_k/|Cl_k(G)|$, results in a runtime of
\begin{equation}
    O\left(\frac{|Cl_k(G)|^2}{\beta_k^2}\frac{n^3 \kappa + n k^2 \zeta_k^{-1/2}}{\delta^2}\right).
\end{equation}

For the LGZ algorithm, a polynomial runtime is maintained only when $\zeta^{-1} \in O(\mathrm{poly}(n))$, which is the case for a graph that is $k$-clique dense. Additionally, it is necessary for $\kappa$ to be polynomial. Although $\lambda_\text{max}$ can be bounded by $O(n)$, there is currently no known lower bound for $\lambda_\text{min}$, meaning it could potentially be exponentially small. It is worth noting that \cite{schmidhuber2022complexity} has demonstrated that both exact and approximate Betti number calculations are NP-hard, even for clique-dense graphs, which implies that it is necessary to investigate alternative definitions of simplicial complexes that are not based on clique complexes.

LGZ is not capable of estimating persistent Betti numbers. Nonetheless, there are other quantum algorithms available for this task, such as the one proposed in \cite{hayakawa2022quantum}. Authors in \cite{mcardle2022streamlined} introduced an alternative method that can achieve a speedup of up to quintic compared to deterministic classical TDA techniques. 

\section{Higher-order topological kernels}

The topological kernel method embeds data into a filtration, extracts topological features, uses to define a distance function, which can immediately be transformed into a kernel. The term \emph{higher-order} denotes that the kernel captures information concerning Betti numbers of orders greater than 1. We denote by $\beta_k(\SComplex)$ the $k$-th Betti number of a given simplicial complex $\SComplex$. 

The Betti curve, also known as Betti sequence, can be used to represent topological information using non-persistent Betti numbers \cite{umeda2017time}. Given a filtration $\filtration: \SComplex_0 \subset ... \subset \SComplex_q$, with $\epsilon_0, ..., \epsilon_q$ the corresponding thresholds, and an integer $k \ge 0$, the Betti curve $\beta_k^\filtration: \R \to \mathbb{N}$ is defined by:
\begin{IEEEeqnarray}{rl}
    \beta_k^\filtration & (\epsilon) = \beta_k(\SComplex_j), \epsilon \in [\epsilon_j, \epsilon_{j+1}).
\end{IEEEeqnarray}
Note that the Betti curve provides less information compared to other topological representations, such as persistence diagrams, which are based on persistent Betti numbers. However, one advantage of the Betti curve representation is that we can exploit LGZ algorithm as a procedure to extract the Betti numbers. Furthermore, authors in \cite{rieck2020topological} have shown that the 1-norm of Betti curves is stable against small perturbations of the dataset. 

The concept of Betti curves can be extended to include multiple orders of Betti numbers $k$, resulting in the \emph{multivariate Betti curve} $\beta_{\le k}^\filtration: \mathbb{N} \times \R \to \mathbb{N}$,
\begin{IEEEeqnarray}{rl}
    \beta_{\le k}^\filtration & (k', \epsilon) = \begin{cases}
    \beta_{k'}(\SComplex_j), & k' \le k \text{ and } \epsilon \in [\epsilon_j, \epsilon_{j+1}) \\
    0, & k' > k \text{ or } \epsilon \not\in [\epsilon_0, \epsilon_q]
    \end{cases}
\end{IEEEeqnarray}
We can represent such information by means of a matrix of real elements $B_{\le k}^\filtration \in \R^{(k+1) \times (q+1)}$,
\begin{equation}\label{eq:bcmatrix}
    (B_{\le k}^\filtration)_{i,j} = \beta_{i}(\SComplex_j); \quad 0 \le i \le k, 0 \le j \le q
\end{equation}

We can establish an upper bound on the error of $\tilde{B}_{\le k}^\filtration$, which is the approximation of $B_{\le k}^\filtration$ obtained by using LGZ to estimate the Betti numbers. To derive this bound, we assume that we are operating in the optimal regime, i.e. $k'$-clique dense graph for all $k' \le k$.
\begin{theorem}
Consider the Vietoris-Rips filtration $\filtration: \SComplex_0 \subset ... \subset \SComplex_q$, $G_0, ..., G_q$ the corresponding graphs, $\epsilon_0, ..., \epsilon_q$ the corresponding thresholds, and an integer $k \ge 0$. Use LGZ to estimate (non-normalized) Betti numbers up to a multiplicative error $\delta > 0$. Then, we can calculate $\tilde{B}_{\le k}^\filtration$, the approximation of $B_{\le k}^\filtration$, in time
\begin{equation}
    O\left(qk \times \frac{|Cl_k(G)|^2}{\beta_k^2}\frac{n^3 \kappa + n k^2 \zeta_k^{-1/2}}{\delta^2}\right).
\end{equation}
with relative Frobenious norm upper bounded by $\delta$,
\begin{equation}\label{eq:relfrobnorm}
    \frac{\lVert 
    \tilde{B}_{\le k}^\filtration - B_{\le k}^\filtration \rVert_F
    }{\lVert  B_{\le k}^\filtration \rVert_F} \le \delta.
\end{equation}
\end{theorem}
\begin{proof}
The runtime is obtained by separately estimating all the elements of the matrix, which adds an overhead proportional to $qk$. The upper bound on the relative Frobenius norm comes from the properties of the norm and elementary arithmetic rules,  
\begin{equation*}
    \lVert 
    \tilde{B}_{\le k}^\filtration - B_{\le k}^\filtration \rVert_F \le \lVert 
    (1+\delta) B_{\le k}^\filtration - B_{\le k}^\filtration \rVert_F \le |\delta| \cdot \lVert B_{\le k}^\filtration \rVert_F.
\end{equation*}
\end{proof}

If we consider a dataset of $m$ elements, where each element corresponds to a cloud of $n$ points, we can place an upper bound on the number of distinct thresholds $q$ as $m(n-1)(n-2)/2$. This is due to the fact that each point can have at most $(n-1)(n-2)/2$ different distances, which can be determined by examining the adjacency matrix of the graph. Since the matrix is symmetric and the diagonal contains all zeros, we only need the values in the upper triangular part without the principal diagonal. In practical scenarios, the actual number of thresholds can be significantly smaller than the upper bound. 

We can define a distance between multidimensional Betti curves, given $B_{\le k}^{\filtration,(1)}, B_{\le k}^{\filtration,(2)}$ defined on the \emph{same} sequence of thresholds $\epsilon_0, ..., \epsilon_q$, 
\begin{IEEEeqnarray}{l}\label{eq:distance}
    \nonumber \mathrm{d}(B_{\le k}^{\filtration,(1)}, B_{\le k}^{\filtration,(2)}) = \\ \left(\sum_{k'=0}^k \sum_{j=0}^{q-1} (\epsilon_{j+1}-\epsilon_j) \Big((B_{\le k}^{\filtration,(1)})_{k', j} - (B_{\le k}^{\filtration,(2)})_{k', j} \Big)^p \right)^{\frac{1}{p}}
\end{IEEEeqnarray}
It should be noted that Equation (\ref{eq:distance}) represents an instance of the weighted Minkowski distance; for $p=2$, it becomes a weighted Euclidean distance.

The concepts of similarity (kernel) and distance measures are closely related. One example is the use of a kernel function $\kappa$ to induce a distance measure between two objects, 
\begin{equation}
    \mathrm{d}(a, b) = \kappa(a, a) + \kappa(b, b) - \kappa(a, b). 
\end{equation}
Conversely, it is also possible to define a kernel function given a distance measure $d$. One common approach is to use the Gaussian kernel, 
\begin{equation}\label{eq:kernel}
    \kappa(B_{\le k}^{\filtration,(1)}, B_{\le k}^{\filtration,(2)}) = \exp(- \gamma \, \mathrm{d}(B_{\le k}^{\filtration,(1)}, B_{\le k}^{\filtration,(2)}))
\end{equation}
where $\gamma > 0$. Here, the distance function $\mathrm{d}$ corresponds to the one defined in Equation (\ref{eq:distance}).

\begin{theorem}[]
The function in (\ref{eq:kernel}) is a Mercer kernel. 
\end{theorem}
\begin{proof}
Follows from the definition of Gaussian kernel \cite{steinwart2008support}. 
\end{proof}

The kernel defined by (\ref{eq:kernel}) is tailored for use with the LGZ algorithm, which does not require calculating persistent Betti numbers, in contrast to the ones previously proposed in the classical machine learning literature \cite{chazal2021introduction,rieck2020topological}. Nevertheless, quantum algorithms capable of estimating persistent Betti numbers have the ability to generate topological kernels based on alternative representations, such as persistence diagrams (as discussed in Sec. \ref{sec:quantum}).

\section{Experimental assessment}

\begin{figure*}[tbp]
    \centering
    \includegraphics[height=0.90\textheight]{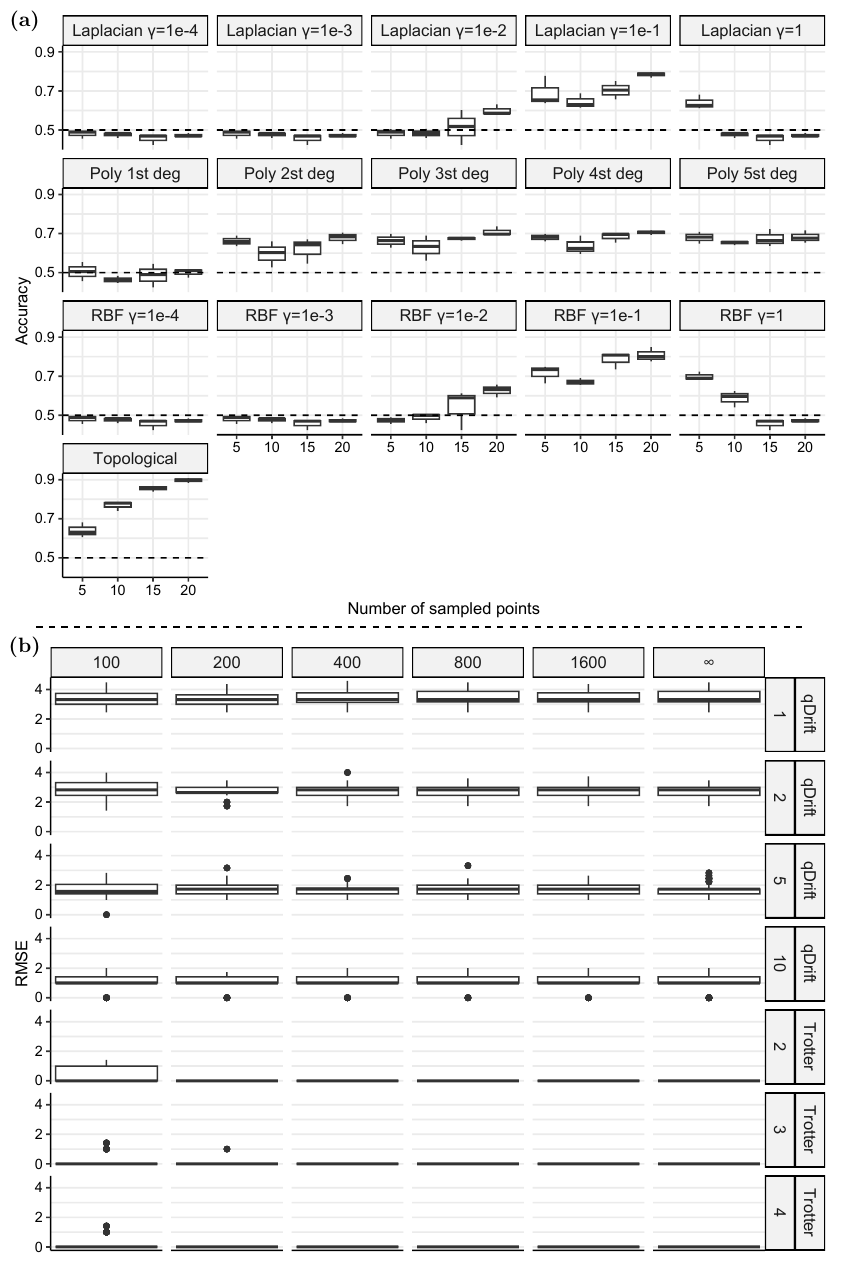}
    \caption{Experimental results. \textbf{(a)} shows the accuracy of kernel machines using various kernels and configurations, plotted against the number of sampled points on the $x$-axis and accuracy on the $y$-axis. \textbf{(b)} shows the root mean square error (RMSE) between the topological kernel presented in (\ref{eq:kernel}), computed through the classical exact procedure, and the LGZ-calculated version with an increasing number of shots (from left to right) and different Hamiltonian simulation techniques, while also varying the number of circuit repetitions (labeled from top to bottom on the rightmost side).}
    \label{fig:quantum_plots}
\end{figure*}

We show a working prototype of our approach and test it over the shape classification problem, a small-scale synthetic benchmark feasible to solve on a quantum simulator. To generate the dataset, we sampled points from the perimeter of various shapes. We then created Vietoris-Rips complexes for different values of $\epsilon$ and extracted the Betti numbers for various values of $k$ to construct the kernel as defined in (\ref{eq:kernel}). After applying this kernel to a Support Vector Machine (SVM), we measured its accuracy and compared it to the accuracy of an SVM that utilizes a conventional classical kernel. The machine learning pipeline used in this experiment is pictured in Fig. \ref{fig:tda-pipeline}(c). 

Furthermore, we investigate the robustness of our approach by examining how much the topological kernel computed using the quantum algorithm deviates from the exact result computed on classical hardware. To test this, we vary the Hamiltonian simulation techniques (Trotter and qDrift), the number of repetitions, and the number of shots used for sampling the results.

\subsection{Setup}

The dataset used for shape classification is artificially created procedurally. The objective of this classification problem is to differentiate between a triangle and a sliced quadrangle, given a set of uniform sample points on their perimeters. It is worth noting that these two shapes are not topologically equivalent and therefore possess distinct topological characteristics. Alternatively, we could differentiate digits zero and eight, topologically equivalent to triangles and sliced quadrangles, for some datasets such as MNIST.

To create each triangle, we described it in $\R^2$ using the points $\{(0,0), (0,1), (1,1) \}$. We then apply an affine transformation that rotates the shape by a randomly generated angle between $0$ and $2\pi$ radians and introduces a skew by randomly generating shear angles along the $x$ and $y$ axes with degrees between $0$ and $\pi/4$. Similarly, we construct each square using the points $\{(0,0), (0,1), (1,1), (1,0) \}$ and define its perimeter as the external border joined with the segment $(0,0) - (1,1)$. We then apply the same transformations as for the triangles to generate various instances of the square. Once generated the dataset of shapes, we sample an increasing number of points from the perimeters of these.

The dataset created has 100 items and is balanced. It has been randomly split into training and testing sets so that each subset is balanced too. We have sampled clouds of points from this dataset, with the number of points ranging from 5 to 20 for topological kernels calculated using classical methods. Furthermore, we have defined a smaller dataset of $m=20$ items and sampled 5 points from each item.

Firstly, we have compared the performance of the topological kernel (\ref{eq:kernel}) on the larger dataset having $\gamma = 1.0, p = 2$ with the Gaussian kernel (hyperparameter $\gamma = 1e{-4}, 1e{-3}, 1e{-2}, 1e{-1}, 1$), Laplacian kernel (hyperparameter $\gamma = 1e{-4}, 1e{-3}, 1e{-2}, 1e{-1}, 1$), and the polynomial kernel (degree $d = 1, 2, 3, 4, 5$). The metric chosen for the comparison is the accuracy with respect to the testing set. In this case, the topological kernel has been calculated classically and the result is exact. 

Secondly, we have used the smaller dataset to tackle the shape classification problem through the use of topological kernels, which are computed using both classical exact methods and a quantum algorithm on a classical simulator. Then, we have compared the difference between the two versions by means of root mean square error (RMSE), which is defined as $\mathrm{RMSE}(K, \tilde{K}) = \sqrt{(\sum_{i,j=1}^m (K)_{i,j} - (\tilde{K})_{i,j})/m^2}$. To implement LGZ, we have used Hamiltonian simulation techniques such as Trotter (with a Trotter number spanning $2, 3, 4$) and qDrift (with a repetition number of $1, 2, 5, 10$), along with varying numbers of shots. This allows us to analyze the impact of these factors on the solution.

\subsection{Results}

In Figure \ref{fig:quantum_plots}(a), we can observe the accuracy of kernel machines utilizing various kernel techniques. Notably, the topological kernel yields the best performance compared to the conventional kernels. This can be explained by the inductive bias of topological data analysis. TDA assumes that the data has an underlying topological structure that can be analyzed and is informative for understanding its properties. This assumption is certainly true for the shape classification problem.

More importantly, the topological technique is the only one that exhibits improvement in performance with an increase in sampled points. This behavior cannot be seen for the other kernels, particularly the RBF kernel which shows the second-best performance. The ability of our kernel to remain robust in relation to the set of sampled points of our figures is especially significant and serves as a strong motivation for utilizing the topological kernel over the other approaches. 

In Figure \ref{fig:quantum_plots}(b), we provide empirical evidence of the robustness of our approach regarding the choice of Hamiltonian simulation technique and the number of shots. It is noteworthy that the performance of the topological kernel created using LGZ closely resembles the exact classical calculations. The qDrift technique introduces stochasticity in the simulation, thus requiring a greater number of repetitions to attain satisfactory performance. This is in contrast to the Trotter technique, which is more resource-intensive, yet yields the best performance (zero RMSE) with only two repetitions. The number of shots needed is relatively unimportant, as 200 shots are sufficient to accurately estimate the coefficients.

\section{Conclusion and future work}

We have introduced a technique for generating a topological kernel using the LGZ quantum algorithm. We have established an upper bound on the error that may arise from the use of this algorithm and demonstrated that the resulting kernel satisfies the criteria of a Mercer kernel. Finally, we have provided a working prototype of our method and tested it on a synthetic benchmark for shape classification. Our approach is feasible when using Trotter and qDrift as Hamiltonian simulation techniques (the latter only with a large enough number of repetitions of the circuit) and with a modest number of shots.

The work we have presented has the potential to be extended in several ways. Firstly, the literature lacks sufficient evidence regarding the impact of higher-order Betti numbers on solving problems. This is due to the previous inability to compute such features before the introduction of quantum algorithms, making the question of their relevance moot. To gain a better understanding of which use cases are suitable for quantum computation, we must analyze various scenarios. To this purpose, one promising area is time series processing, which can be efficiently encoded as a point cloud in a $d$-dimensional space using Takens' embedding. Tuning the parameter $d$ to generate simplicial complexes with the highest order of Betti numbers could potentially enable us to reach a regime in which the quantum algorithm performs optimally. Time series has several applications in biology, engineering, and finance, including fraud detection \cite{dipierro2021quantum}. 

Secondly, an area of further exploration could be the potential combination of the proposed topological kernel with conventional (Gaussian, Laplacian, polynomial) kernels. This approach has the potential to leverage the flexibility of kernel methods and utilize the property that the sum, product, and limits of a sequence of Mercer kernels are also Mercer kernels. Consequently, it may be possible to create a hybrid kernel that incorporates both geometrical and topological information.

Thirdly, the skeleton of the Vietoris-Rips complex is constructed using a distance function, typically Euclidean. However, it is possible to substitute this distance function with one induced by a quantum kernel. Doing so may enable us to uncover relationships between data that are not apparent using traditional methods. An example of the potential advantages of these quantum kernels has been demonstrated in the context of physics-related problems  \cite{wozniak2023quantum}.

Fourthly, all the tested methods have the potential to be enhanced by optimizing the kernel parameters. This involves adjusting the parameter $\gamma > 0$ for the Gaussian and Laplacian kernels, the degree $d$ for the polynomial kernel, while for the topological kernel, both the $\gamma > 0$ parameter and the maximum order of Betti number $k$ might be optimized. It is also possible to enhance the approach by restricting the multivariate Betti curves, represented in matrix form in Equation (\ref{eq:bcmatrix}), to a randomly sampled subset of the thresholds $\epsilon_0, ..., \epsilon_q$.

Fifthly, the stability of the multivariate Betti curve can be intuitively inferred from the proof of the stability of the univariate Betti curve presented in \cite{rieck2020topological}, but this has to be formally proven.  

\section*{Code and data availability}

The code and data used in this study are available upon request by contacting the authors for access. 

\section*{Acknowledgment}

We acknowledge the CINECA award under the ISCRA initiative, for the availability of high-performance computing resources and support. MI thanks Oriel Kiss for the insightful discussion.


\end{document}